\documentclass[conference]{IEEEtran}
\ifCLASSINFOpdf
\else
\fi

\usepackage{amssymb}
\usepackage{amsmath}
\usepackage{graphicx}
\usepackage{color}
\usepackage{algorithm}
\usepackage{algpseudocode}
\usepackage{comment}

\usepackage{amsthm}
\newtheorem{proposition}{Proposition}

\newtheorem{lem}{Lemma}

\newtheorem{definition}{Definition}
\usepackage{tikz}
\usetikzlibrary{arrows,calc,positioning,shadows,shadows.blur}

\tikzstyle{intg}=[draw,rounded corners,minimum size=4em,text centered,text width=2.2cm,drop shadow,fill=white,text height=0.3cm]
\tikzstyle{intg2}=[draw,rounded corners,minimum size=4em,text centered,text width=1.9cm,drop shadow,fill=white,text height=0.48cm]
\tikzstyle{intg3}=[draw,rounded corners,minimum size=4em,text centered,text width=1.6cm,drop shadow,fill=white,text height=0.48cm]

\IEEEoverridecommandlockouts

\IEEEpubid{
   \makebox[\columnwidth]{ISBN 978-3-903176-08-9~\copyright~2018 IFIP \hfill}
   \hspace{\columnsep}
   \makebox[\columnwidth]{ }
}

\begin{document}
%
\title{Profit and Strategic Analysis for MNO-MVNO Partnership}

\author{\IEEEauthorblockN{Nesrine Ben Khalifa, Amal Benhamiche, Alain Simonian, Marc Bouillon}
\IEEEauthorblockA{Orange Labs, France\\
Email:  firstname.name@orange.com}
}


%


\maketitle
\begin{abstract}
We consider a mobile market driven by two Mobile Network Operators (MNOs) and a new competitor Mobile Virtual Network Operator (MVNO). 
The MNOs can partner with the entrant MVNO by leasing network resources; however, the MVNO can also rely on other technologies such as free WiFi access points. Moreover, in addition to its connectivity offer, the MVNO can also draw indirect revenues from services due to its brand. In that framework including many  access technologies and several revenue sources, a possible partner MNO will then have to decide which wholesale price to charge the MVNO for its resources. This multi-actor context, added to the need to consider both wholesale and retail markets, represents a new challenge for the underlying decision-making process. In this paper, the optimal price setting is formulated as a multi-level optimization problem which enables us to derive closed-form expressions for the optimal MNOs wholesale prices and the optimal MVNO retail price. The price attractivity of the MVNO is also evaluated in terms of its indirect revenues and the proportion of resources leased from possible partner MNOs. Finally, through a game-theoretical approach, we characterize the scenario where both MNOs partner with the MVNO as the unique Nash equilibrium under appropriate conditions. 
\end{abstract}
%

%
\IEEEpeerreviewmaketitle

\section{Introduction}
Mobile telecommunication markets are usually covered by a few number of operators because of high infrastructure and spectrum license costs. In addition, mobile operators are facing the challenges of upgrading their networks to 5G technology in order to cope with the increasing demand of users' traffic. This context, as well as the virtualization of wireless networks, may lower the barrier for the entrance  of Mobile Virtual Network Operators (MVNO) to the market. While not possessing 
their own network infrastructure, MVNOs can lease capacities from Mobile Network Operators (MNOs) on the wholesale market to ensure wireless services to their customers.


A new generation of MVNOs has recently emerged, capable of leasing resources from different MNOs while also taking advantage of the available WiFi hotspots in order to propose a full connectivity offer to their customers. The latter can therefore afford to blindly use multiple network/technology services to establish the communications and use mobile Internet without any specific setting. Typically, such MVNOs can launch their activity thanks to specific inter-operator partnerships with MNOs for the cellular infrastructure utilization while the WiFi offloading part remains a unilateral decision. A recent example is offered by Google through its Project-Fi \cite{googleFi}, launched in the U.S. in partnership with three leading MNOs, namely Sprint, T-Mobile and U.S. Cellular.    
This illustrates the case of a competitive MVNO able to absorb a significant part of the retail market (on mobile and data services) from the MNOs without being an expert in networking or even possessing the physical infrastructures and the ability to manage them. 

On the other hand, an entrant MVNO can be already positioned at a higher place in the Value Chain and have income from its new customers through its current ``non telco'' activity. In fact, this MVNO may have a good reputation as an OTT ("Over the Top") providing content or other high level services; it can thus draw significant additional revenue, hereafter termed as \textit{indirect} to differentiate it from the telco offer, from its new customers.  

This new multi-technological and multi-activity context therefore raises new questions as to the interactions between MNOs and such a MVNO, namely the optimal price setting adopted by all actors and the consequences on the respective market shares and profits. The economic viability of a possible partnership between MNOs and the MVNO should, in particular, be understood on account of the existence of alternative technologies and other sources of revenues.

\subsection{State-of-the-art} 
The ecosystem of MNO-MVNO relationships has been addressed in many studies. A detailed description of the MVNOs' classes in terms of their dependence to the host operator is presented in \cite{liau}, where it is shown that a MVNO can be classified as either {\it{light}} or {\it{full}}. The authors describe, in particular, the possible business models of MVNOs and examine the impact of different parameters such as the MNO's market share on the outcome of cooperation between the host and virtual operators. 
In \cite{banerjee}, the authors analyze the incentives for MNOs to form a strategic cooperation with MVNOs. They specifically examine the effects of the brand appeal of the MVNO and the wholesale discount offered by the MNO on the fulfillment of mutually beneficial partnerships. 

Besides, the economic viability of MNO-MVNO relationship has been investigated in \cite{infocom17, net_economics} and references therein. First in \cite{infocom17}, it is argued that the business of a MVNO may be profitable in a transitory phase when it partners with MNOs with a small market share; however, it is shown in \cite{net_economics} that, in the long term, the MVNO is better off when it preferably partners with a big MNO, say, the incumbent. The latter point of view of a stabilized market and mature economic actors will be adopted in the sequel.

A multi-stage game for modeling the MNO-MVNO interaction is presented in \cite{mvno_2017} where the MNO investment, the MVNO's decision on the leasing from the MNO and the retail pricing are successively focused on at each stage. In \cite{globecom_2011}, spectrum leasing and pricing are studied with game-theoretic models. A comprehensive study of the market share between MNO and MVNO based on the brand appeal of the MNO is provided in \cite{Debbah2012} where the authors also use game theory tools. 

A recent study of Google-Fi like MVNOs is provided in \cite{wiopt2017} where the price setting between multiple service providers and the virtual operator is examined. In that paper, 
the user defection rate to the MVNO is assumed to be simply constant with no account of the impact of the MVNO price on its customer base; further, the authors only optimize the MVNO price for given wholesale MNO prices. 

In contrast, we will here consider a more accurate economic model wherein (1) using the so-called price-demand elasticity, the users reply to the MVNO offer depends on the varying \textit{price difference} between the MVNO offer and that of the other MNOs; (2) beside the search of an optimal retail price for the MVNO, we also determine the wholesale MNO prices by \textit{maximizing their respective profit}.

\subsection{Addressed issues and contributions} 
In this paper, we make a thorough economic analysis of strategic MNO-MVNO partnership. We consider the upcoming of a new MVNO, a potential competitor proposing low-cost services and threatening the MNO market share. Specifically, we address the following questions:
\begin{itemize}
\item[$\bullet$] When a MNO decides to conclude a partnership with the MVNO, what is the best price setting for the partner MNO in order to maximize its profit? 
\item[$\bullet$] When the wholesale prices are fixed by partner MNOs, what is the optimal price that should be charged by the MVNO to its customers ?
\item[$\bullet$] What is the impact of the MVNO's indirect revenues on its optimal retail price ?
\item[$\bullet$] What is the best decision for the MNOs facing the entry of the MVNO?
\end{itemize}
In this aim, the optimal price setting is addressed via a Stackelberg Game involving leaders and followers \cite{net_economics} (Section 2.3.6). In particular, we define and study two different decision-making models for the optimal wholesale price setting, namely a {\it{fully sequential}} model and a {\it{partially sequential}} model; we then determine the optimal retail price of the MVNO. Furthermore, we study the impact of the MVNO's indirect revenue on its optimal retail price in both decision-making models. Finally, we propose a game-theoretical approach to determine the Nash equilibrium under sufficient conditions on this indirect revenue and discuss its economic interpretation.
\subsection{Paper structure}
In Section \ref{model}, we formally introduce the economic model describing the interactions between MNOs and the MVNO. In Section \ref{pricing}, we formulate the price setting problem and study the wholesale and retail price optimization for all actors. A game-theoretical setting is discussed in Section \ref{nash}. In Section \ref{numerical}, we comment our general results on economic grounds. Finally, some concluding remarks are given in Section \ref{conc}.
\section{Economic Model and Assumptions} \label{model}
We consider a mobile operator market composed of two MNOs that share the whole customer base. The upcoming of a new MVNO proposing an attractive price 
impacts the repartition of customers who can defect from their original operators to the benefit of this MVNO. The MNOs must then identify the best decisions they are able to make, that is, to decide whether to partner with the MVNO. An operator that would decide to contract a partnership with the MVNO would obviously not be immune to lose customers but might, nevertheless, recover a part of the lost retail revenue via the wholesale income. 
Given this ecosystem, the economic variables of the MNOs and MVNO activity can be described as follows:

{\bf{MNO's profit:}} the profit of a MNO is the sum of its revenues obtained from the retail market for end-users (and possibly from the Business-to-Business wholesale market for partner MNOs) from which are subtracted the corresponding costs. Denote by $p_{i}$ (resp. $w_{i}$) the unit price of MNO $i$ on the retail (resp. the wholesale) market; the retail revenues then depend linearly on the operator's customer base, while the wholesale revenues depend linearly on the amount of MVNO traffic accommodated by the MNO's network. Besides, we consider unit \textit{network costs} $c_i$, \textit{non-network costs} $\widetilde c_i$ (IT, commercial, etc.) and \textit{fixed costs} 
$\overline{C}_i$;  as a linear approximation, the total network and non-network costs are assumed to depend linearly on the customer base whereas fixed costs are independent of the amount of the operator's customer base;  



{\bf {MVNO's profit:}} the total revenue of the MVNO is the sum of \textit{direct} revenues obtained from the retail market (with net unit price $p_{0}$) and of \textit{indirect} revenues (with unit price 
$r_{0}$) obtained from e.g. advertising. As to network costs, we here consider that the MVNO enables its users' devices either to automatically connect to a free public WiFi Access Point (AP) or to a partner MNO's network, depending on the Quality of Service of each access mode. The proportion of MVNO traffic handled by free WiFi APs is denoted by $\gamma \in [0,1[$; the only network costs of the MVNO are therefore those caused by the other proportion $1 - \gamma$ of traffic dealt with MNOs through partnerships. In addition to these network costs, the MVNO finally incurs non-network and fixed costs denoted by $\widetilde c_0$ and $\overline{C}_0$, respectively.

For both MNOs, we assume that long-distance (backhaul and core) network costs are negligible compared to the access costs; the possible wholesale offer of any MNO will therefore mainly account for the transportation through its cellular access network of MVNO traffic.     

\begin{table}[t]\caption{Key Terms and Symbols} 
\begin{tabular}{|c|l|}
\hline
{\bf{Symbol}} &{{\bf{Definition} ($i$ = 1, 2)}} \\ \hline
$p_i$, $c_i$, $\widetilde c_i$& Unit retail price, network and non-network costs \\
 &             of MNO $i$ \\ \hline
$\overline{C}_i$ & Fixed costs of MNO $i$ \\ \hline
$Q_i$ & Customer base of MNO $i$ before the MVNO's entry \\ \hline
$Q$ & Total customer base \\ \hline
$\pi_i= Q_i/Q$ & Market share of MNO $i$ before the MVNO's entry \\  \hline
$Q_{i,0}$ &  Customer base of MNO $i$ which defects to the MVNO \\ \hline
$Q_0$ & Total customer base of the MVNO \\ \hline
$r_0$ & Net unit indirect revenues of the MVNO  \\ \hline
$p_0$, $\widetilde c_0$  & Unit retail price, non-network costs of the MVNO \\ \hline
$\overline{C}_0$ & Fixed costs of the MVNO \\ \hline
$w_i$ & Unit wholesale price charged by MNO $i$ to the MVNO \\ \hline
$\gamma$  & Proportion of MVNO traffic handled by \\
&           free WiFi APs \\ \hline
\end{tabular}\label{tab:notation} 
\end{table}

{\bf{Users Behavior:}} the users' reply to the MVNO's offer is assumed to depend only on the relative price of that offer. Specifically, the behavior of users is modeled according to the price-demand elasticity \cite{economics} so that the part of MNO $i$ customers which defects to the MVNO is expressed by
\begin{eqnarray}\label{defection}
Q_{i,0} = \varepsilon \, Q_i \left ( \frac{p_i-p_0}{p_i} \right )
\end{eqnarray}
where $Q_i$ denotes the customer base of MNO $i$ before the MVNO joins the market and $\varepsilon > 0$ is the price-demand elasticity coefficient (although generally depending on prices, $\varepsilon$ is here assumed to be a constant on the basis of a small price variability range). Without loss of generality, we assume throughout the paper that $p_{2}\leqslant p_{1}$ and $p_{0}\leqslant p_{2}$; this ensures that both $Q_{1,0}$ and $Q_{2,0}$ are non-negative. The MVNO's total customer base is then
\begin{eqnarray}\label{Q0}
Q_0=Q_{1,0}+Q_{2,0}.
\end{eqnarray}

{\bf{MVNO traffic split:}} in the case when both MNOs partner with the MVNO, the MVNO traffic which is not supported by free WiFi APs is split between MNOs networks proportionally to their market share before the MVNO's upcoming. Consequently, the traffic transported for the MVNO on MNO $i$ network equals $(1 - \gamma)\pi_{i}Q_{0}$, where $\pi_i = Q_i/Q$ is the market share of MNO $i$ before the MVNO's entry (with $\pi_1 + \pi_2 = 1$). This is motivated by the fact that the MVNO will partner with a MNO all the more that the latter has a large market share.  

Note finally that a sample value of parameter $\gamma$ is given by the ratio of the duration spent on WiFi access to the overall duration of a given communication session; this ratio can then be averaged over all successive sessions to provide the mean proportion $\gamma$; the latter is clearly related to the given geographic density of the WiFi APs. 

\begin{table}[t]\caption{MVNO Traffic Split}
\centering
\begin{tabular}{|c|l|}
\hline
{\bf{Traffic Amount}} & {\bf{Accommodation}} \\ \hline
$\gamma Q_0$ & WiFi APs\\ \hline
$(1-\gamma) \pi_1 Q_0$ & MNO 1's Cellular BSs\\ \hline
$(1-\gamma)\pi_2 Q_0$ & MNO 2's Cellular BSs\\
\hline
\end{tabular} \label{tab:traffic_2}
\end{table}

Tables \ref{tab:notation} and \ref{tab:traffic_2} sum up the notation used in the paper.
\section{Wholesale and Retail Price Setting} 
\label{pricing}
In this section, we address the optimization of MNOs and MVNO profits in order to determine the optimal wholesale and retail prices. Specifically, we introduce several optimization problems where $w_{i}$, $i \in \{1,2\}$, and $p_{0}$ are the decision variables. In the sequel, we denote by "Part" the strategy consisting in partnering with the MVNO and by "NonPart" the strategy consisting in not partnering with the MVNO. 
In a competitive environment where the MNOs do not collude with each other, two scenarios can be envisaged:
\begin{itemize}
\item[$\mathbf{A.}$] Only one operator, either MNO $1$ or MNO $2$, decides to contract with the MVNO. We denote this scenario by (Part-NonPart) or (NonPart-Part), respectively;
\item[$\mathbf{B.}$] Both operators choose to contract with the MVNO. We denote this scenario by (Part-Part).
\end{itemize}
These two scenarios are successively analyzed below.
\subsection{\textbf{Scenario (Part-NonPart)}}
\label{sec:part-nonpart}
Consider first the case when only MNO $i$ proposes a wholesale offer to the MVNO, while MNO $-i$ decides not to partner with the new entrant (by convention, $i$ = 1 or 2 implies $-i$ = 2 or 1). In this case, the traffic $\gamma Q_0$ generated by the MVNO's users will be delivered through free WiFi APs and the remaining $(1-\gamma)Q_0$ will be handled by the partner MNO's network. The MVNO's profit is therefore given by
\begin{eqnarray}\label{u_0}
{\cal{R}}_0(p_0,w_i)=(p_0+r_0)Q_0-w_i (1-\gamma) Q_0-\widetilde c_0 Q_0-\overline{C}_0.
\end{eqnarray}
The optimal MVNO's retail price $p_0^*$ can then be determined by solving the following optimization problem:
\begin{eqnarray}\label{prob_0}
&&\underset{{0 \leqslant p_0 \leqslant p_2 }}{\mbox{ max }} {\cal{R}}_{0}(p_0,w_i).
\end{eqnarray}

\begin{lem}\label{lem:p0star}
In the (Part-NonPart) scenario, given the wholesale price $w_{i}$, the optimal MVNO retail price equals
\begin{eqnarray}\label{eq:p0star}
p_0^*(w_i)=\min(\widetilde p_0(w_i), p_2)
\end{eqnarray}
where
\begin{eqnarray}\label{p0tilde}
\widetilde p_0(w_i)=\frac{1-\gamma}{2}w_i+ \frac{Q}{2\mathcal{S}}+\frac{\widetilde c_0-r_0}{2},
\end{eqnarray}
with $\mathcal{S}= Q_1/p_1+ Q_2/p_2$.
\end{lem}
\begin{proof}
First observe that, in view of definitions \eqref{defection} and \eqref{Q0}, $Q_{0}$ involved in expression \eqref{u_0} is a linear function of variable $p_{0}$, thus making the profit $\mathcal{R}_{0}$ a quadratic function of $p_{0}$. Given the price $w_{i}$, the first order optimality condition ${\partial {\cal{R}}_0} (p_0,w_i) /{\partial p_0}=0$ 
for problem (\ref{prob_0}) then provides the critical point $\widetilde p_0(w_i)$ as given in \eqref{p0tilde}.
Besides, the second derivative $\partial^2{\cal{R}}_0(p_0,w_i)/\partial p_0^2$ is equal to the negative constant $-2\varepsilon\mathcal{S}$; ${\cal{R}}_0(p_0,w_i)$ is therefore a strictly concave function of $p_0$ with a unique maximum at price $p_0^*(w_i)$ given by \eqref{eq:p0star}. 
\end{proof}

Now, we consider the profit of MNO $i$ given by
\begin{align}
{\cal{R}}_i(p_0,w_i) = & \, p_i(Q_i-Q_{i,0})+w_i (1-\gamma)Q_0 \; - 
\nonumber \\
& \, c_i(Q_i- Q_{i,0}+(1-\gamma)Q_0) \; - 
\nonumber \\
& \, \widetilde c_i(Q_i-Q_{i,0}) - \overline{C}_i,
\nonumber
\end{align}
that is, 
\begin{eqnarray}\label{eq:Ri}
{\cal{R}}_{i}(p_0,w_i)= h_i(Q_i-Q_{i,0})+(w_i-c_i)(1-\gamma)Q_0-\overline{C}_i
\end{eqnarray}
where $h_i=p_i-c_i-\widetilde c_i \geqslant$ 0. The partner MNO $i$ seeks to maximize its profit ${\cal{R}}_i$. To this end, we replace $p_0$ involved in expression \eqref{eq:Ri} via $Q_{i,0}$ and $Q_0$ by its optimal value 
$p_0^*(w_i)$ derived in Lemma \ref{lem:p0star}; in fact, the MNO anticipates the best pricing strategy of the MVNO and thus sets its optimal wholesale price based on this anticipation. Define then
\begin{eqnarray}\label{eq:Rii}
\mathcal{R}_i^*(w_i)={\mathcal{R}}_i(p_0^*(w_i),w_i)
\end{eqnarray}
which is obtained through \eqref{eq:Ri} with $Q_{i,0}=\varepsilon Q_i (p_i-p_0^*(w_i))/p_i$ and $Q_0$ given by \eqref{Q0}. The optimization problem for MNO $i$ can then be expressed by
\begin{eqnarray}\label{prob_i}
&&\underset{{w_i \geqslant 0}}{\mbox{ max }} \mathcal{R}_i^*(w_i). 
\end{eqnarray}

\begin{proposition}\label{prop_one_fi}
In the (Part-NonPart) scenario, the optimal MNO's wholesale price equals
\begin{eqnarray}\label{w_star}
\widehat{w}_i= \min (\overline{w}_i, \widetilde w_i)
\end{eqnarray}
where we set
\begin{eqnarray}\label{w_bar}
\overline{w}_i=\frac{1}{1-\gamma} \left ( 2 p_2-\frac{Q}{\mathcal{S}}+r_0-\widetilde c_0 \right )
\end{eqnarray}
and 
\begin{eqnarray}\label{w_tilde}
\widetilde w_i= \frac{c_i}{2}+ \frac{1}{1-\gamma} \left ( \frac{h_i Q_i}{2p_i\mathcal{S}}+\frac{Q}{2\mathcal{S}}+\frac{r_0-\widetilde c_0}{2} \right ).
\end{eqnarray}
The optimal MVNO's retail price is then determined by
\begin{eqnarray}\label{p0tilde_final}
\widehat p_0(\widehat{w}_i)=\min(\widetilde p_0(\widehat{w}_i), p_2).
\end{eqnarray}
Defining the constant 
\begin{eqnarray}\label{ri0bar}
\overline{r}_{i,0}=\frac{h_iQ_i}{p_i\mathcal{S}} + c_i(1 - \gamma) + \frac{3Q}{\mathcal{S}} + \widetilde{c}_{0} - 4p_2,
\end{eqnarray}
we then have $\overline{w}_i \leqslant \widetilde w_i \iff p^{*}_{0}(\widehat{w}_i) = p_2 \iff r_0 \leqslant \overline{r}_{i,0}$. 
\end{proposition}
\begin{proof}
We consider the two cases $(a)$ $\widetilde p_0 (w_i) > p_2$ and 
$(b)$ $\widetilde p_0 (w_i) \leqslant p_2$. First consider case $(a)$. This corresponds to values of $w_i$ such that $w_i > \overline{w}_i$ where $\overline{w}_i$ is given by \eqref{w_bar}. Relation \eqref{eq:p0star} then yields $p_0^*(w_i) = p_2$  and by \eqref{eq:Rii}, we easily show that $\mathcal{R}_i^*(w_i) = {\mathcal{R}}_{i}(p_2,w_i)$ is a linear and increasing function of $w_i$. The optimal value of $\mathcal{R}_i^*$ is thus obtained at $w_{i} = +\infty$; but this unbounded price giving the value $-\infty$ for $\mathcal{R}_0$, case $(a)$ is thus eventually excluded.

Now consider case $(b)$. This corresponds to values of $w_i$ such that $w_i \leqslant \overline{w}_i$. Relation \eqref{eq:p0star} now yields $p_0^*(w_i) = \widetilde p_0 (w_i)$; using \eqref{eq:Rii} again and writing now $$Q_{i,0}=\varepsilon \frac{Q_i}{p_i}(p_i-\widetilde p_0 (w_i)), \quad Q_0=Q_{1,0}+Q_{2,0}$$
as functions of $w_i$, $\mathcal{R}_i^*(w_i)={\mathcal{R}}_i(\widetilde p_0(w_i),w_i)$ can then be easily expressed as a quadratic function of $w_i$. 
The first order optimality condition  $\partial \mathcal{R}_i^*(w_i)/\partial w_i=0$ for problem (\ref{prob_i}) yields the critical point $\widetilde w_i$ as given in (\ref{w_tilde}).
Besides, the second derivative $\partial^2\mathcal{R}_i^*(w_i)/\partial w_i^2$ is equal to the negative constant $-\varepsilon (1-\gamma)^2\mathcal{S}/2$. Therefore, $\mathcal{R}_i^*(w_i)$ is a strictly concave function of $w_i$ with a unique maximum at $\widehat{w}_i$ given by \eqref{w_star}.

The optimal MVNO's retail price is obtained by replacing $w_i$ in \eqref{eq:p0star} by $\widehat{w}_i$ given in (\ref{w_star}), hence (\ref{p0tilde_final}). Finally, elementary algebra reduces condition $\overline{w}_i \leqslant \widetilde w_i$ to $r_0 \leqslant \overline{r}_{i,0}$, with $\overline{r}_{i,0}$ given by (\ref{ri0bar}).
\end{proof}
All previous results symmetrically hold for the (NonPart-Part) scenario.

\subsection{\textbf{Scenario (Part-Part)}}
\label{subsection_two_partners}
We now turn to the situation where both MNOs partner with the MVNO. Two models can be proposed depending on the order in which decisions are taken, namely:   
\begin{itemize}
\item[$\bullet$] A \textit{Fully Sequential} (FS) model where a leader MNO (say, the one with the highest market share) first chooses its wholesale price; then the second MNO, the follower, determines its wholesale price accordingly; finally, the MVNO chooses its retail price. This situation is illustrated in Fig. \ref{sequential}-{\it{Left}}.
\item[$\bullet$] A \textit{Partially Sequential} (PS) model where first the two MNOs (say, with comparable weights) choose their respective wholesale price without coordination; then, the MVNO chooses its retail price. This situation is illustrated in Fig. \ref{sequential}-{\it{Right}}.
\end{itemize}
\begin{figure}[t!]
\centering
\begin{tikzpicture}[>=latex',auto]
\node [intg] (kp) {Leader MNO $i$ chooses $w_i$};
\node [intg] (ki33)[node distance=1.8cm,below of=kp] {MNO $-i$ chooses $w_{-i}$};
\node [intg] (ki4) [node distance=1.8cm,below of=ki33] {MVNO chooses $p_0$};
\draw[->,draw=black,fill=white,line width=1pt,] (kp) -- (ki33);
\draw[->,draw=black,fill=white, line width = 1pt] (ki33) -- (ki4);
\end{tikzpicture} \hskip -0.1cm \begin{tikzpicture}[>=latex',auto]
\node [intg2] (kp2)[right = 0.1cm of  kp ]  {MNO $1$ chooses $w_1$};
\node [intg2] (ki3) [node distance=2.5cm,right of=kp2] {MNO $2$ chooses $w_{2}$};
\node [intg] (ki4) [node distance=2.9cm ,below of=kp2, right = -1.0 cm of kp2 ] {MVNO chooses $p_0$};
\draw[->,draw=black,fill=white,line width=1pt,] ($(ki3.south)+(-0.05,0)$)  -- ($(ki4.north)+(0.65,0)$);
\draw[->,draw=black,fill=white,line width=1pt,] ($(kp2.south)+(-0.01,0)$)  -- ($(ki4.north)+(-0.6,0)$);
\draw[->,draw=white,fill=white,line width=1pt,] ($(kp2.south)+(-1.5,0)$)  -- ($(ki4.north)+(-0.6,0)$);
\end{tikzpicture}
\caption{Hierarchical decision models. {\it {Left, }} Fully sequential  model. {\it{Right, }} Partially sequential model.}
\label{sequential}
\end{figure}
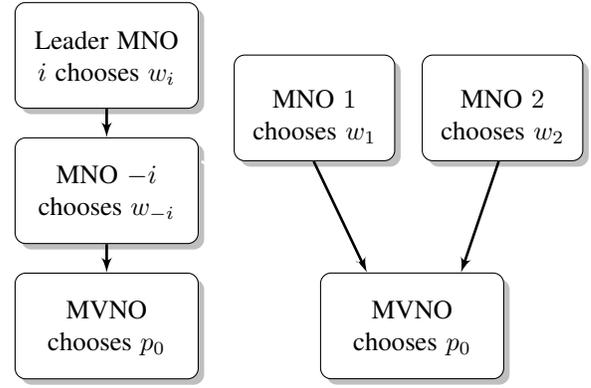 
In both models, we consider that the MNOs move before the MVNO because they own the resources  to lease to the latter. They forecast the MVNO's reply to their pricing strategies and choose the wholesale prices that maximize their profits, given the anticipated MVNO's optimal pricing strategy. The MVNO eventually chooses its retail price, given the wholesale prices fixed by the MNOs.

The MVNO's profit is now given by
\begin{eqnarray}\label{eq:R0FiFi}
{\cal{R}}_0(p_0,w_1,w_2)&=&(p_0+r_0)Q_0 - w_1 (1-\gamma)\pi_1Q_0- \nonumber \\ &&w_2 (1-\gamma) \pi_2 Q_0-\widetilde c_0 Q_0-\overline{C}_0
\end{eqnarray}
and the optimization problem of the MVNO is formulated by
\begin{eqnarray}\label{prob_0_fi_fi}
&&\underset{{0 \leqslant p_0 \leqslant p_2 }}{\mbox{ max }} {\cal{R}}_{0}(p_0,w_1,w_2). 
\end{eqnarray}
\begin{lem}
\label{lem:p0star_fi_fi}
In the (Part-Part) scenario, given the wholesale prices $w_1$ and $w_2$, the optimal retail price of the MVNO equals
\begin{eqnarray}\label{p0star_fi_fi}
p_0^*(w_1,w_2)=\min(\widetilde p_0(w_1,w_2),p_2)
\end{eqnarray}
where
\begin{eqnarray}\label{p0tilde_fi_fi}
\widetilde p_0(w_1,w_2)=\frac{(1-\gamma)}{2}\left(\pi_1 w_1 + \pi_2 w_2\right) + \frac{Q}{2\mathcal{S}} +\frac{\widetilde c_0 - r_0}{2}.
\end{eqnarray}
\end{lem}
\begin{proof}
Similar to that of Lemma \ref{lem:p0star}.
\end{proof}
In this section, the notations for profit $\mathcal{R}_i$ should not be confused with that of Section \ref{sec:part-nonpart}.
Now, consider the profit of MNO $i\in\{1,2\}$ given by
\begin{eqnarray*}
{\cal{R}}_i(p_0,w_i)&=&p_i(Q_i-Q_{i,0})+w_i (1-\gamma) \pi_i Q_0 -\nonumber \\ &&  c_i(Q_i-Q_{i,0} +  (1-\gamma) \pi_i Q_0)- \nonumber \\ &&  \widetilde c_i(Q_i-Q_{i,0})-\overline{C}_i,
\end{eqnarray*}
that is, 
\begin{align}
\label{eq:Ri_Fi_Fi}
{\cal{R}}_{i}(p_0, w_i) = & \; h_{i} (Q_{i}-Q_{i,0}) \; +
\\ \nonumber
& \; (w_{i}-c_{i})(1-\gamma) \pi_{i} Q_0-\overline{C}_{i}
\end{align}
where we set $h_{i}=p_{i}-c_{i}-\widetilde c_{i}$. Both MNO partners 1 and 2 seek to maximize their own profit ${\cal{R}}_1$ and ${\cal{R}}_2$. In order to solve this optimization problem for either model (FS) or (PS), we replace $p_0$ by its optimal value $p_0^*(w_1, w_2)$ derived above in Lemma \ref{lem:p0star_fi_fi}. Define then 
\begin{eqnarray}\label{eq:RiiFiFi}
\mathcal{R}_{i}^*(w_1,w_2) = {\mathcal{R}}_i(p_0^*(w_1,w_2),w_i)
\end{eqnarray}
as obtained from \eqref{eq:Ri_Fi_Fi} with $Q_{i,0}=\varepsilon Q_i (p_i-p_0^*(w_1, w_2))/p_i$ and $Q_0$ given by \eqref{Q0}. 

We now successively address the maximization of MNOs profits for the (FS) and (PS) models.
\vspace{0.1in}
\subsubsection{\textbf{Fully Sequential Model}} 
Assume that MNO $i \in \{1,2\}$ is the leader and MNO $-i$ is the follower (by convention, $-i = 2$ if 
$i = 1$, and $-i = 1$ if $i = 2$). Given $w_i$, the follower thus decides on the wholesale price $w_{-i}$ to charge the MVNO. The optimization problem for both MNOs can then be expressed by the following bilevel formulation
\begin{equation}\label{pb:FS}
\left\{
\begin{array}{l}
\underset{w_i\geqslant 0}{\mbox{ max }}\mathcal{R}_{i}^*(w_1, w_2)_{\vert w_{-i} = w_{-i}^*}, \\ \\
\mbox{subject to} 
\quad\quad w_{-i}^* = \underset{w_{-i}\geqslant 0}{\mbox{argmax }}\mathcal{R}_{-i}^*(w_1, w_2)
\end{array} \right.
\end{equation}
where the notation $\vert w_{-i} = w_{-i}^*$ means that function $\mathcal{R}_i^*$ is evaluated for the variable $w_{-i}$ equal to $w_{-i}^*$; note that $w_{-i}^*$ is a function of $w_{i}$. The symmetrical case when MNO $-i$ is the leader is similarly defined. In order to solve problem (\ref{pb:FS}), we introduce the following definitions.

\begin{definition}
\label{def:delta}
We denote by $\Delta$ the closed triangular region defined by 
$\Delta$ = $\{(w_1, w_2)\in\mathbb{R}^+\times\mathbb{R}^+: \widetilde p_0(w_1, w_2)\leqslant p_2\}$ 
where $\widetilde p_0(w_1, w_2)$ is given by \eqref{p0tilde_fi_fi} (see Fig. \ref{fig:delta}).
   
\vspace{0.1in}
Further denote by $\delta$ the boundary segment of $\Delta$ defined by 
$\delta$ = $\{(w_1, w_2)\in\Delta: \widetilde p_0(w_1, w_2) = p_2\}$. 
\end{definition}
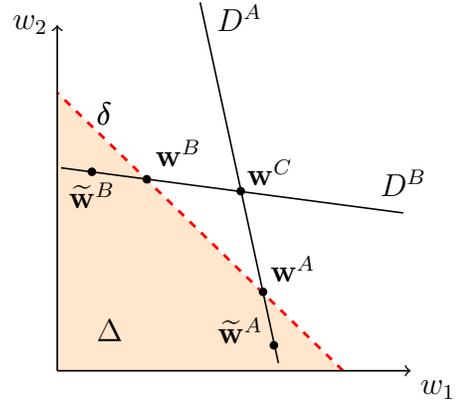
\begin{figure}[!ht] 
\centering
\begin{tikzpicture} 
\draw[draw=white,fill=orange!20!white] (0,0) -- (3.8,0) -- (0,3.7) -- cycle; 
\draw[dashed,line width=1.05 pt,draw=red] (3.8,0) -- (0,3.7);
\foreach \coordinate/\label/\pos in {(0.4,0.8)/{{\large{$\Delta$}}}/below right,(2.4,2.34)/{{\large{$\mathbf{w}^C$}}}/above right,(2.71,1.06)/{{\large{$\mathbf{w}^A$}}}/above right,(1.18,2.6)/{\large{$\mathbf{w}^B$}}/above right,(0.4,3.7)/{\large{$\delta$}}/below right,(2.42,4.4)/{\large{$D^A$}}/above,(0.4,3.7)/{\large{$\delta$}}/below right,(4.6,2.2)/{\large{$D^B$}}/above,(2.85,0.24)/{\large{$\widetilde{\mathbf{w}}^A$}}/above left,(0.46,2.63)/{\large{$\widetilde{\mathbf{w}}^B$}}/below}
\node[\pos] at \coordinate {\label};
\draw[line width=0.6pt, draw=black](2.94,0.1) -- (1.9,4.9) ;%
\draw [fill] (2.439,2.389) circle [radius=0.05];
\draw [fill] (2.736,1.05) circle [radius=0.05];
\draw [fill] (1.19,2.548) circle [radius=0.05];
\draw [fill] (2.88,0.339) circle [radius=0.05]; 
\draw [fill] (0.46,2.647) circle [radius=0.05]; 
\draw[line width=0.6pt](4.6,2.1) -- (0.05,2.7) ;
\draw[->, line width=0.64pt] (0.0,0) -- (4.7,0) node[below right] {{\large{$w_1$}}} coordinate(x axis); 
\draw[->, line width=0.64pt] (0,0.0) -- (0,4.6) node[left] {{\large{$w_2$}} }coordinate(y axis); 
\end{tikzpicture} 
\caption{The region $\Delta$, its boundary $\delta$, and the points $\mathbf{w}^A$, $\mathbf{w}^B$, 
$\widetilde{\mathbf{w}}^A$, $\widetilde{\mathbf{w}}^B$ and $\mathbf{w}^C$ (lines $D^A$ and $D^B$ are invoked in Appendix \ref{seq:Prop2}).} 
\label{fig:delta}
\end{figure}
\begin{definition}
\label{def:wAwB}
Let $\mathbf{w}^A=(w_1^A, w_2^A)$ and $\mathbf{w}^B=(w_1^B, w_2^B)$ denote the pair of prices given by
$$
\left\{
\begin{array}{ll}
w_{1}^A = \displaystyle \frac{\displaystyle \frac{h_1 Q_1}{p_1}+c_1 (1-\gamma)  \pi_1 \mathcal{S}+2Q-2\mathcal{S} p_2}{(1-\gamma)\pi_1 \mathcal{S}}, \nonumber \\
\\
w_{2}^A = \displaystyle \frac{\displaystyle -\frac{h_1 Q_1}{p_1}-c_1  (1-\gamma)\pi_1 \mathcal{S}-4Q+4 p_2 \mathcal{S}+\mathcal{T}}{(1-\gamma)\pi_2 \mathcal{S}}
\end{array}\right.
$$
and
$$
\left\{
\begin{array}{ll}
w_{1}^B =  \displaystyle \frac{\displaystyle -\frac{h_2 Q_2}{p_2}-c_2 (1-\gamma)\pi_2  {\cal{S}}-4Q+4 p_2 {\cal{S}}+\mathcal{T}}{(1-\gamma)\pi_1{\cal{S}}}, \nonumber \\
\\
w_{2}^B =  \displaystyle \frac{\displaystyle \frac{h_2Q_2}{p_2}+c_2 (1-\gamma)\pi_2{\cal{S}}+2Q-2 {\cal{S}} p_2}{(1-\gamma)\pi_2 {\cal{S}}},
\end{array}\right.
$$
respectively, with $\mathcal{T} = Q+(r_0-\widetilde c_0) \mathcal{S}$ for short. 

Define also the function $\Omega_{-i}$, $i \in \{1,2\}$, by 
\begin{eqnarray}
\label{eq:omega_i}
\Omega_{-i}(w_{i}) &=&\frac{c_{-i}}{2}+\frac{h_{-i}Q_{-i}}{2p_{-i}(1-\gamma)\pi_{-i}\mathcal{S}}+\frac{Q}{2(1-\gamma)\pi_{-i}\mathcal{S}}- \nonumber \\ && \frac{\widetilde c_0-r_0}{2(1-\gamma)\pi_{-i}}-\frac{\pi_i}{2\pi_{-i}}w_i
\end{eqnarray} 
and the points $\widetilde{\mathbf{w}}^B = (\widetilde w_1, \Omega_{2}(\widetilde w_1))$,  
$\widetilde{\mathbf{w}}^A = (\Omega_{1}(\widetilde w_2), \widetilde w_2)$ where
\begin{eqnarray}\label{eq:witilde_i}
\widetilde w_{i}=\frac{\displaystyle \frac{h_{i} Q_{i}}{p_{i}}-\frac{h_{-i} Q_{-i}}{p_{-i}}+\mathcal{S}(1 - \gamma)( c_{i} \pi_{i} - c_{-i} \pi_{-i}) + \mathcal{T}}{2(1-\gamma) \pi_{i} \mathcal{S}}.
\end{eqnarray}
\end{definition}
\noindent
Note that $\mathbf{w}^A\in\delta$ and $\mathbf{w}^B\in\delta$; a geometric interpretation of the pairs $\mathbf{w}^A$, $\mathbf{w}^B$ and $\widetilde{\textbf{w}}^A$, $\widetilde{\textbf{w}}^B$ is given in Appendix \ref{seq:Prop2}. We can now state the following.
\begin{proposition}\label{prop_FS}
Define the constant $\overline{r}_{0}$ by
\begin{equation*}
\overline r_0=\frac{h_1Q_1}{p_1 \mathcal{S}}+\frac{h_2Q_2}{p_2\mathcal{S}}+(1-\gamma) (\pi_1 c_1+\pi_2 c_2)+\frac{7 \, Q}{\mathcal{S}}+\widetilde c_0-8 p_2.
\end{equation*} 
In the (Part-Part) scenario with FS model,
\begin{itemize}
\item[$\bullet$] if $r_0\leqslant\overline r_0$ and MNO 1 (resp. MNO 2) is the leader, then the optimal wholesale price vector $(w_1^*, w_2^*)$ is given by $\textbf{w}^B \in \delta$ (resp. 
$\textbf{w}^A \in \delta$) introduced above. In either case, the optimal MVNO's retail price is then
$$
p_0^*(w_1^*,w_2^*)= p_2;
$$
\item[$\bullet$] if $r_0 > \overline r_0$ and MNO 1 (resp. MNO 2) is the leader, then the optimal wholesale price vector $(w_1^*, w_2^*)$ is given by $\widetilde{\mathbf{w}}^B \in \Delta \setminus \delta$ (resp. $\widetilde{\mathbf{w}}^A \in \Delta \setminus \delta$) introduced above. In either case, the optimal MVNO's retail price is then
$$
p_0^*(w_1^*,w_2^*)= \widetilde p_0(w_1^*,w_2^*)< p_2.
$$  
\end{itemize}
\end{proposition}
\noindent
We defer the detailed proof to Appendix \ref{seq:Prop2}.  
\vspace{0.1in}
\subsubsection{\textbf{Partially Sequential Model}} Now assume that MNOs simultaneously choose their optimal wholesale prices. Given (\ref{eq:RiiFiFi}), the joint optimization problem for both MNOs is thus expressed by
\begin{eqnarray}\label{max_11}
&&\underset{ w_1 \geqslant 0}{\mbox{ max }} {{\cal{R}}_{1}^*}(w_1,w_2)
\end{eqnarray}
and
\begin{eqnarray}\label{max_22}
&&\underset{ w_2 \geqslant 0}{\mbox{ max }} {{\cal{R}}_{2}^*}(w_1,w_2).
\end{eqnarray}
\begin{definition}\label{def:wC}
Let ${\bf{w}}^C=(w_1^C,w_2^C)$ denote the pair of prices given by
\begin{equation*}
\left \{
\begin{array}{l}
w_{1}^C = \displaystyle \frac{\displaystyle \frac{2h_1 \, Q_1}{p_1} - \frac{h_2 \, Q_2}{p_2}+(1-\gamma)(2c_1 \pi_1-c_2 \pi_2) \mathcal{S}+\mathcal{T}}{3(1-\gamma)\pi_1\mathcal{S}}, 
\nonumber \\ \\
w_{2}^C = \displaystyle \frac{\displaystyle \frac{2h_2 \, Q_2}{p_2} - \frac{h_1 \, Q_1}{p_1}+(1-\gamma)(2c_2 \pi_2-c_1 \pi_1) \mathcal{S}+\mathcal{T}}{3 (1-\gamma)\pi_2 \mathcal{S}}
\end{array}
\right.
\end{equation*} 
with $\mathcal{T} = Q+(r_0-\widetilde c_0) \mathcal{S}$. 
\end{definition}
\noindent
The pair ${\bf{w}}^C$ is given a geometric interpretation in Appendix \ref{seq:Prop3}. This enables us to state the following.

\begin{proposition}\label{prop_PS}
Define the constant 
$r_0^\flat$ by
$$
r_0^\flat=\frac{h_1Q_1}{p_1 \mathcal{S}}+\frac{h_2Q_2}{p_2\mathcal{S}}+(1-\gamma) (\pi_1 c_1+\pi_2 c_2)+\frac{5 \, Q}{\mathcal{S}}+\widetilde c_0-6 p_2.
$$
In the (Part-Part) scenario with PS model,
\begin{itemize}
\item[$\bullet$] if $r_0>r_0^\flat$, the optimal wholesale price vector $(w_1^*,w_2^*)$ is given by ${\bf{w}}^C\in\Delta\setminus\delta$. The optimal MVNO's retail price is then
$$
p_0^*(w_1^C,w_2^C)=\widetilde p_0(w_1^C,w_2^C)<p_2
$$
as defined by (\ref{p0tilde_fi_fi});
\item[$\bullet$] if $r_0 = r_0^\flat$, the optimal wholesale price vector $(w_1^*,w_2^*)$ is given by ${\bf{w}}^C\in\delta$. The optimal MVNO's retail price is then $p_0^*=p_2$;

\item[$\bullet$] if $r_0 < r_0^\flat$, the problem \eqref{max_11}-\eqref{max_22} admits no solution.
\end{itemize}
\end{proposition}

\noindent 
The proof of Proposition \ref{prop_PS} is detailed in Appendix \ref{seq:Prop3}. 
\section{Game-Theoretic Model} \label{nash}
In this section, we propose a non-cooperative game-theoretical model to formalize the competition between MNO 1 and MNO 2. 
\begin{definition}\label{def:game}
Introduce the two-player game where
\begin{itemize}
\item[$\--$] the players are  MNO $1$ and MNO $2$,
\item[$\--$] the strategies are either to contract with the new entrant, "Part" strategy, or not to contract, "NonPart" strategy,
\item[$\--$] the payoffs are given by the matrix
\end{itemize}
\begin{center}
$
\bordermatrix{~     &  Part  &  NonPart \cr
       Part    & ({\cal{R}}_1^*(w_1^*,w_2^*),{\cal{R}}_2^*(w_1^*,w_2^*))& ({\cal{R}}_1^*(\widehat w_1),{\cal{R}}_2(\widehat w_1)) \cr
       NonPart &  ({\cal{R}}_1(\widehat w_2),{\cal{R}}_2^*(\widehat w_2))  &  ({\cal{R}}_1^{0}, {\cal{R}}_2^{0}) \cr}.
$
\end{center}
\end{definition}

\noindent 
In this matrix, the pair $({\cal{R}}_1^{0}, {\cal{R}}_2^{0})$ corresponds to the scenario when no MNO partners with the MVNO, thus forbidding its entrance into the market. 

We recall that a Nash Equilibrium (NE) is a strategy profile such that no player has an incentive to unilaterally deviate from this profile \cite{Nash}. Besides, recall from definition (\ref{defection}) that $Q_{i,0}$ denotes the customer part which defects from MNO $i$ to the new entrant in the scenario (Part, Part); furthermore, let $\widehat{Q}_{i,0}$ denote the customer part which defects from MNO $i$ to the MVNO in the scenario (NonPart, Part), that is, when only MNO $-i$ contracts a partnership with the new entrant.
\begin{lem}\label{lemma3}
In the FS model with $r_0\leqslant \overline{r}_0$, a MNO loses more customers when it is non-partner than when it is a partner of the MVNO, that is, $\widehat{Q}_{i,0} \geqslant Q_{i,0}$.
\end{lem}
\begin{proof}
From Proposition \ref{prop_FS} with $r_0\leqslant \overline{r}_0$, we have $p_0^*(w_1^*,w_2^*)=p_2$ while Proposition \ref{prop_one_fi} entails $\widehat p_0(\widehat{w}_{-i}) \leqslant p_2$. Using definition \eqref{defection}, we then deduce that the difference 
$Q_{i,0}-\widehat{Q}_{i,0} = \frac{\varepsilon Q_i}{p_i} \left ( \widehat{p}_0(\widehat{w}_{-i})- p_0^*(w_1^*,w_2^*) \right )$ is non-positive, as claimed.
\end{proof}
As a consequence, we can formulate the subsequent result on the existence of a Nash equilibrium.
\begin{proposition}\label{prop:nash}
In the FS model with $r_0\leqslant \overline{r}_0$ and for wholesale prices higher than network costs, 

(a) the scenario (Part, Part) is a NE; 

(b) if $r_0 \leqslant \min(\overline{r}_0,\overline{r}_{2,0})$, (Part, Part) is the unique NE.
\end{proposition}
\begin{proof}
\textit{(a) Scenario (Part, Part) is a NE}. Suppose that both operators partner with the MVNO, that is, the scenario is (Part, Part); consider MNO $i$ and let ${\cal{R}}_i^*(w_1^*,w_2^*)$ denote its profit for this scenario. Assume now that MNO $i$ unilaterally switches from strategy "Part" to strategy "NonPart"; we denote by ${\cal{R}}_i(\widehat w_{-i})$ the profit of MNO $i$ when it applies strategy "NonPart" while MNO $-i$ keeps strategy "Part". We then have
\begin{eqnarray*}
{\cal{R}}_i^*(w_1^*,w_2^*)=h_i(Q_i-Q_{i,0})+(w_i^*-c_i)(1-\gamma)\pi_iQ_0-\overline{C}_i
\end{eqnarray*}
and ${\cal{R}}_i(\widehat w_{-i})=h_i(Q_i-\widehat{Q}_{i,0})-\overline{C}_i$, so that 
\begin{align}
{\cal{R}}^*_i(w_1^*,w_2^*)-{\cal{R}}_i(\widehat w_{-i}) & \, = h_i({\widehat{Q}}_{i,0}-Q_{i,0}) + (w_i^*-c_i) \times 
\nonumber \\
& \, (1-\gamma) \pi_i Q_0.
\label{ProfitDiff}
\end{align}
As $\widehat{Q}_{i,0} \geqslant Q_{i,0}$ by Lemma \ref{lemma3}, the profit difference ${\cal{R}}^*_i(w_1^*,w_2^*)-{\cal{R}}_i(\widehat w_{-i})$ in (\ref{ProfitDiff}) is then non-negative after assumptions $h_i \geqslant 0$ and $w_i^* - c_i \geqslant 0$; MNO $i$ then has no incentive to unilaterally deviate from strategy "Part" since this yields a profit decrease. (Part, Part) is thus a NE.

\textit{(b) (Part, Part) is the unique NE}. Assume both MNOs use strategy "Non Part". Now suppose MNO 2 unilaterally switches to "Part"; we first have ${\cal{R}}_2^{0}=h_2Q_2-\overline{C}_2$ while ${\cal{R}}^{*}_{2}(\widehat{w}_2)=h_2(Q_2-\widehat{Q}_{2,0}) + (\widehat{w}_2 - c_2)(1 - \gamma)Q_{0} - \overline{C}_2$ after (\ref{eq:Ri}), hence
\begin{equation}
{\cal{R}}^*_2(\widehat{w}_2) - {\cal{R}}_2^{0} = -h_2\widehat{Q}_{2,0} + (\widehat{w}_2 - c_2)(1 - \gamma)Q_{0}. 
\label{ProfitDiffBIS}
\end{equation}
From Proposition \ref{prop_one_fi}, we have $\widehat{Q}_{2,0}$ = 0 if $r_0 \leqslant \overline{r}_{2,0}$; the difference ${\cal{R}}^*_2(\widehat{w}_2) - {\cal{R}}_2^{0}$ in (\ref{ProfitDiffBIS}) is then  non-negative after assumption $\widehat{w}_2 - c_2 \geqslant 0$ and (Non Part, Non Part) is not a NE. We conclude that if $r_0 \leqslant\overline{r}_0$ and $r_0 \leqslant \overline{r}_{2,0}$, (Part,Part) is the only NE.
\end{proof}

Note that a preliminary study has given us hints for (Part,Part) to be still a NE for 
$r_0 > \overline{r}_0$. For the PS model, a similar analysis should also be addressed for 
$r_0 \leqslant r_0^{\flat}$. 
\section{Economic Discussion} 
\label{numerical}
The results obtained in the previous sections allow us to provide the following comments.

{\bf {Scenario (Part-Part): }} When both MNOs partner with the MVNO, Propositions \ref{prop_FS} and \ref{prop_PS} entail that, when $\gamma \rightarrow $ 1, all optimal wholesale prices are of order $1/(1 - \gamma)$; this increasing rate again confirms the effect of economy of scale. Now, regarding the MVNO's optimal retail price, we distinguish two cases depending on the order in which decisions are taken by MNOs. First, we have shown (Proposition \ref{prop_FS}) for the FS model that the entrant MVNO charges its users the same price $p_2$ as the lowest MNO if $r_{0} \leqslant \overline{r}_0$, thus attracting users only from the MNO with the highest retail price; otherwise, it attracts users from both MNOs. Second, we have shown (Proposition \ref{prop_PS}) for the PS model that the entrant MVNO cannot set an optimal retail price strictly lower than both MNOs' prices unless it has sufficiently high indirect revenues, that is, $r_0 > r_0^\flat$. Otherwise, if $r_0 < r_0^\flat$, there is no wholesale prices that jointly optimize both MNOs profits.    

For each class of actors, we can conclude the following: 

{\bf{For the MVNO:}} for all scenarios, the MVNO retail price $p_0^*$ decreases with $r_0$; the MVNO can thus set a retail price strictly lower than that of the MNOs if it has high enough indirect revenues. Besides, the threshold value $\overline r_0$ (resp. $r_0 ^\flat$) in model FS (resp. model PS) is a decreasing function of the proportion $\gamma$, so that the MVNO has an optimal retail price strictly lower than that of both MNOs for large enough $\gamma$. The technological independence of the MVNO from its partner MNOs due to free WiFi access thus translates into an economic advantage on the retail market (but this does not obviously account for better QoS and security levels offered to its customers if a larger part of MVNO traffic were transferred through optimized cellular networks);

{\bf{For the MNOs:}} for the FS model and under sufficient conditions on the indirect revenue of the MVNO, the scenario (Part, Part) defines the unique NE. This means that the MNOs have then an incentive to partner with the new entrant: in fact, the MNOs would in the first place prefer that the MVNO does not enter the market in order to keep their customer base, but each MNO fears that its competitor hosts the MVNO, in which case it would incur losses both on the retail and wholesale markets. As a consequence, both MNOs will eventually decide to partner with the MVNO. In addition, a non-partner MNO would incur higher retail losses if it were non-partner than when it is a partner of the MVNO. Cooperating with the MVNO therefore enables each MNO to compensate for a part of its retail revenue losses.
\section{Conclusion} 
\label{conc}
In this paper, our contribution is twofold. First, we address the \textit{price setting optimization  problems} for both MNOs and the entrant MVNO in the framework of two distinct scenarios. In this aim, we propose several mathematical programming formulations for the underlying problems, each one corresponding to a specific decision-making scheme; we also discuss the economic interpretation of the optimal solution in each case. Secondly, based on the optimal price setting step, we provide a \textit{game-theoretical analysis} of the MNOs competition and show that (Part, Part) is the unique Nash Equilibrium, and thus the most profitable scenario for both MNOs, provided that appropriate conditions on the MVNO's indirect revenue are fulfilled. 

The particular case of only two competing MNOs has provided us with interesting results, and a natural extension to this work would be to generalize the results obtained for two MNOs to an arbitrary number ($n \geqslant$ 3) of MNOs. Indeed, the two-dimensional optimization problems that we have here addressed may exhibit other features in the $n$-dimensional case (e.g. several possible optimal points). Furthermore, the associated $n$-player games could be amenable to cooperation schemes among players which could be interestingly studied. On the other hand, other demand models differing from that considered in this paper (customer defection due to the price/demand elasticity) could also be envisaged; alternative models based on interactions between users or on the MVNO brand appeal could capture other preference sources of users towards each actor.   


\section*{Acknowledgment}
The authors thank V. Chandrakumar, L. Le Beller and M. Touati at Orange Labs for fruitful discussions, together with the anonymous referees for their valuable comments. 

\bibliographystyle{IEEEtran}

\section{Appendix}
\subsection{\textbf{Proof of Proposition \ref{prop_FS}.}}
\label{seq:Prop2}
Recall that MNO $i \in \{1,2\}$ (resp. MNO $-i$) is assumed to be the leader (resp. the follower). We successively consider the two cases $(a)$ $(w_1, w_2)\notin\Delta$ and $(b)$ $(w_1, w_2)\in\Delta$.

First consider case $(a)$. In view of \eqref{p0star_fi_fi}, this corresponds to values of $(w_1,w_2)$ such that $p_0^*(w_1,w_2)=p_{2}$ and by \eqref{eq:RiiFiFi}, we easily show that ${\cal{R}}_{-i}^*(w_1, w_2) = {\cal{R}}_{-i}(p_2,w_{-i})$ is a linear and increasing function of $w_{-i}$. The optimal value of $\mathcal{R}_{-i}^*$ is thus obtained at $w_{-i} = +\infty$; as this unbounded price gives the value $-\infty$ for $\mathcal{R}_0$, case $(a)$ can thus be excluded.

Now consider case $(b)$. By \eqref{p0star_fi_fi}, this corresponds to pairs $(w_1,w_2)$ such that 
$p_0^*(w_1,w_2)=\widetilde p_{0}(w_1, w_2)$. To solve the optimization problem for follower MNO $-i$ in (\ref{pb:FS}) with given $w_i$, we introduce the associated Lagrange function given by
\begin{eqnarray*}
{\cal{L}}_{-i}(w_1,w_2,\lambda_{-i})={\cal{R}}_{-i}^*(w_1,w_2)-\lambda_{-i}\left (\widetilde p_0(w_1,w_{2})- p_2 \right )
\end{eqnarray*}
where $\lambda_{-i}$ is the Lagrange multiplier associated to the constraint $\widetilde p_0(w_1,w_2) \leqslant p_2$; from definition \eqref{eq:RiiFiFi} and the expression \eqref{p0tilde_fi_fi} of 
$\widetilde p_0(w_1,w_2)$, ${\cal{R}}_{-i}^*(w_1,w_2)$ is easily expressed as a quadratic function of variable $w_{-i}$. The system of Karush-Kuhn-Tucker (KKT) (\cite{nonlinear_optimization}, Chap.4, Sec. 4.2.13) conditions for the Lagrangian ${\cal{L}}_{-i}$ above can be written as
\begin{equation}\label{kkt}
\left \{
\begin{array}{l}
\displaystyle \frac{\partial {\cal{L}}_{-i}}{\partial w_{-i}}(w_1, w_2,\lambda_{-i})=0, \\ \\
\lambda_{-i} \geqslant 0, \qquad \lambda_{-i} \big( \widetilde p_0(w_1,w_2)-p_2 \big)=0, \\ \\
\widetilde p_0(w_1,w_2) \leqslant p_2.
\end{array}
\right.
\end{equation}
Two cases can intervene for the multiplier $\lambda_{-i}$:

(I) if $\lambda_{-i} =0$, the first KKT condition in system \eqref{kkt} reads 
\begin{equation}\label{DADB}
\frac{\partial \mathcal{R}_{-i}^*}{\partial w_{-i}}(w_1,w_2)=0;  
\end{equation}
function $\mathcal{R}^*_{-i}$ being quadratic, equation \eqref{DADB} is linear in both variables $w_1$, 
$w_2$ and thus defines geometrically a line $D$ (displayed in Fig. \ref{fig:delta} as line $D^B$ if $-i = 2$ or line $D^A$ if $-i = 1$). Solving \eqref{DADB} for $w_{-i}$ then yields the unique maximum at point 
$w_{-i}^*=\Omega_{-i}(w_{i})$ with function $\Omega_{-i}$ defined as in \eqref{eq:omega_i}.

Now, consider the profit maximization in problem \eqref{pb:FS} for the leader MNO $i$. In order to solve it, we replace $w_{-i}$ in the profit $\mathcal{R}^*_i(w_1,w_2)$ by the maximum $w_{-i}^*=\Omega_{-i}(w_i)$ derived above. Given \eqref{eq:RiiFiFi}, we thus define
\begin{eqnarray*}
{\mathcal{R}}_{i}^{**}(w_{i})=\mathcal{R}_i^*( w_{1},w_2)_{\vert w_{-i} = \Omega_{-i}(w_i)}.
\end{eqnarray*}
Recall that we consider case $(b)$ for which $(w_1,w_2)\in\Delta$ and $p_0^*(w_1,w_2)=\widetilde p_0(w_1,w_2)$, so that the optimization problem for leader MNO $i$ eventually reads  
\begin{equation}
\left\{
\begin{array}{ll}
\underset{ w_i \geqslant  0 }{ \max~  } {\cal{R}}_{i}^{**}(w_i), 
\\ \\
\mbox{subject to } \quad \widetilde p_0(w_1,w_2)_{\vert w_{-i} = \Omega_{-i}(w_i)}\leqslant p_2.
\end{array} \right.
\label{u2222}
\end{equation}
First, the constraint $\widetilde p_0(w_1,w_2)_{\vert w_{-i} = \Omega_{-i}(w_i)}\leqslant p_2$ in \eqref{u2222} is easily translated into $w_{i}\leqslant \overline{w}_{i}$ where we set
\begin{eqnarray*}
\overline{w}_{i}=
\frac{\displaystyle -\frac{h_{-i}Q_{-i}}{p_{-i}}-c_{-i} (1-\gamma)\pi_{-i}  \mathcal{S} - 
4Q + 4 p_2 \mathcal{S} + \mathcal{T}}{(1-\gamma)\pi_{i} \mathcal{S}}
\end{eqnarray*}
with $\mathcal{T} = Q + (r_0 - \widetilde{c}_0)\mathcal{S}$. Second, the 1st order condition $\partial \mathcal{R}_{i}^{**}(w_i)/\partial w_i =0$ for problem 
\eqref{u2222} yields the critical point $\widetilde w_i$, given as in \eqref{eq:witilde_i}. The second derivative $\partial^2 {\cal{R}}_{i}^{**}(w_i)/\partial w_i^2$ being equal to the negative constant 
$-\varepsilon  (1-\gamma)^2 \pi_i^2 \mathcal{S}/2$, $\mathcal{R}_{i}^{**}(w_i)$ is therefore a strictly concave function of $w_i$ and has a unique maximum on $\mathbb{R}^+$ at $\widetilde w_{i}$. 
It thus follows from the latter discussion that
$$
w_i^*=\min (\widetilde w_{i},\overline{w}_{i})
$$
is the unique solution to problem \eqref{u2222}. 
Now, we easily verify that $\widetilde w_{i} \geqslant \overline{w}_{i} \iff $  
$r_0 \leqslant \overline{r}_{0}$ where $\overline{r}_{0}$ is expressed in Proposition \ref{prop_FS}. We thus conclude that if $r_0 \leqslant \overline{r}_{0}$, the optimal solution is the intersection point 
$\textbf{w}^*$ = $(\overline{w}_{1}, \Omega_{2}(\overline{w}_1))$ = $\textbf{w}^B \in D^B \cap \delta$ 
when $i=1$, or $\textbf{w}^* = (\Omega_{1}(\overline{w}_2), \overline{w}_{2})$ = $\textbf{w}^A \in D^A \cap \delta$ when  
$i=2$; otherwise, if $r_0 > \overline{r}_{0}$, the optimal solution is the intersection point 
$\textbf{w}^*$ = $(\widetilde{w}_{1}, \Omega_{2}(\widetilde{w}_1)) = \widetilde{\textbf{w}}^B \in D^B \cap \Delta$ when $i=1$, or $\textbf{w}^*$ = $(\Omega_{1}(\widetilde{w}_2),\widetilde{w}_{2}) = \widetilde{\textbf{w}}^A \in D^A \cap \Delta$ when $i=2$;.

(II) otherwise, if $\lambda_{-i} > 0$, the first KKT condition 
$\partial {\cal{L}}_{-i}(w_1,w_2,\lambda_{-i})/\partial w_{-i}=0$ in system (\ref{kkt}) yields
\begin{eqnarray}\label{eq_ww}
w_{-i}=\bar{\omega}_{-i} (\lambda_{-i},w_{i}),
\end{eqnarray}
with $\overline{\omega}_{-i} (\lambda_{-i}, w_{i})=A_i(w_i)/\varepsilon \mathcal{S}  (1-\gamma) \pi_{-i}$
where we set 
\begin{align}
A_i(w_i) = & \; \varepsilon \frac{h_{-i} Q_{-i}}{2p_{-i}}+\varepsilon \frac{Q}{2}-\varepsilon \mathcal{S}w_{i}\frac{1-\gamma}{2}\pi_{i}-\varepsilon \mathcal{S} \frac{\tilde c_0-r_0}{2} \; + 
\nonumber \\ 
& \; \varepsilon c_{-i} (1-\gamma) \pi_{-i} \frac{\mathcal{S}}{2}-\frac{\lambda_{-i}}{2}.
\nonumber
\end{align}
The complementary slackness condition $\widetilde p_0(w_{1},w_{2})=p_2$ then eventually reduces to
$\lambda_{-i}=\Lambda_{-i}(w_{i})$ where
\begin{align}
\Lambda_{-i}(w_{i}) = & \; \varepsilon \mathcal{S}  (1-\gamma) \pi_{i}  w_{i} + 
\varepsilon \frac{h_{-i}Q_{-i}}{p_{-i}} + 3 \varepsilon Q \; + 
\nonumber \\ 
& \; \varepsilon \mathcal{S}(\tilde c_0-r_0) + \varepsilon \mathcal{S} c_{-i} (1-\gamma)  \pi_{-i} - 
4 \varepsilon \mathcal{S} p_2.
\nonumber
\end{align}
By inserting this expression of $\lambda_{-i} = \Lambda_{-i}(w_{i})$ into the right-hand side of (\ref{eq_ww}), we finally get $w_{-i}=\overline{\Omega}_{-i}(w_{i})$ where
$$
\overline{\Omega}_{-i}(w_{i})=\frac{-(1-\gamma) \pi_{i}  \mathcal{S}w_{i} -Q-\mathcal{S}(\tilde c_0-r_0)+2\mathcal{S} p_2}{\mathcal{S} (1-\gamma) \pi_{-i} }.
$$
Now, consider the profit of leader MNO $i$ given by
\begin{eqnarray*}
\mathcal{R}_{i}^{**}(w_{i})=\mathcal{R}_i^*(w_1,w_2)_{\vert w_{-i} = \overline{\Omega}_{-i}(w_i)}.
\end{eqnarray*}
The first derivative of ${\cal{R}}_{i}^{**}(w_i)$ with respect to $w_{i}$ equals
$d {\cal{R}}_i^{**}(w_{i})/d w_{i}=\pi_{i} (1-\gamma) Q_0$ and is strictly positive; but in the present case, $w_i\in [0, w_i^0]$ where $w_i^{0}$ is either the abcissa or the ordinate of the intersection point of line $\delta$ with the axis $w_{-i}$ = 0. The corresponding optimal unit price $w_{-i}^{*}$ for the follower MNO $-i$ is therefore 0; as having no economic relevance for this partner MNO, this case (II) is eventually excluded.
\subsection{\textbf{Proof of Proposition \ref{prop_PS}.}}
\label{seq:Prop3}
We successively consider the two cases $(a)$ $(w_1, w_2) \notin \Delta$ and $(b)$ $(w_1, w_2) 
\in \Delta$.

First consider case $(a)$. In view of \eqref{p0star_fi_fi}, this corresponds to pairs $(w_1,w_2)$ such that $p_0^*(w_1,w_2)=p_{2}$ and by \eqref{eq:RiiFiFi}, we easily show that $\mathcal{R}_i^*(w_1, w_2) = {\mathcal{R}}_{i}(p_2,w_{i})$ is a linear and increasing function of $w_{i}$, $i \in \{1, 2\}$. The optimal value of $\mathcal{R}_i^*$ is thus obtained at 
$w_{i} = +\infty$; as this unbounded price gives the value $-\infty$ for $\mathcal{R}_0$, case $(a)$ can therefore be excluded.

Now consider case $(b)$. In view of \eqref{p0star_fi_fi}, this corresponds to values of $(w_1,w_2)$ such that $p_0^*(w_1,w_2)=\widetilde p_{0}(w_1, w_2)$. From definition \eqref{eq:RiiFiFi} and the expression \eqref{p0tilde_fi_fi} of $\widetilde p_0(w_1,w_2)$, each function ${\cal{R}}_1^*(w_1,w_2)$ and ${\cal{R}}_2^*(w_1,w_2)$ is again expressed as a quadratic function of variable $w_1$ and $w_2$, respectively. The KKT conditions associated with optimization problem (\ref{max_11}) for MNO $1$, and optimization problem (\ref{max_22}) for MNO $2$, read
\begin{equation}\label{kktPSi}
\left \{
\begin{array}{l}
\displaystyle \frac{\partial {\cal{R}}_{i}^*}{\partial w_i}(w_1,w_2)=\lambda_i \frac{ (1-\gamma) \pi_i}{2}, \\ \\
\lambda_i \geqslant 0, \qquad \lambda_i (\widetilde p_0(w_1,w_2)-p_2)=0,
\\ \\
\widetilde p_0(w_1,w_2) \leqslant p_2;
\end{array}
\right.
\end{equation}
for $i = 1$ and $i = 2$, respectively. 
At this point, we need to distinguish three cases according to the values of the Lagrange multipliers (I) $\lambda_1=0$ and $\lambda_2=0$, (II) $\lambda_1>0$ and $\lambda_2>0$, (III) $\lambda_1=0$ and $\lambda_2>0$ (or reversely $\lambda_1 > 0$ and $\lambda_2 = 0$).

(I) First assume $\lambda_1=0$ and $\lambda_2=0$. The simultaneous conditions
$$
\frac{\partial {\cal{R}}_{1}^*}{\partial w_1}(w_1,w_2)=0, \quad 
\frac{\partial {\cal{R}}_{2}^*}{\partial w_2}(w_1,w_2)=0
$$
yield the critical pair ${\bf{w}}^C \in D^A\cap D^B$, intersection point of lines $D^A$ and $D^B$ defined in \eqref{DADB}. For each $i \in \{1,2\}$, the second derivative $\partial^2\mathcal{R}^*_i(w_1,w_2)/\partial w_i^2$ equals the negative constant $-\varepsilon (1-\gamma)^2 \pi_i^2 \mathcal{S}$; $\mathcal{R}^*_i$ is thus a strictly concave function of the variable $w_i$ and the point $\mathbf{w}^C$ is therefore the unique joint maximum for $\mathcal{R}^*_1$ and $\mathcal{R}^*_2$. Besides, it is readily shown that $\widetilde p_0(w_1^C,w_2^C) < p_2$ if and only if $r_0 > r_0^\flat$; in such a case, we then have 
${\bf{w}}^C \in \Delta \setminus \delta$ and this point ${\bf{w}}^C$ is the optimal pair of wholesale prices.

(II) Now assume $\lambda_1>0$ and $\lambda_2>0$. Solving each KKT system (\ref{kktPSi}) for $i = 1$ and $i = 2$, we get
\begin{equation*}
\left \{
\begin{array}{l}
w_1(\lambda_1)= \displaystyle \frac{\displaystyle \frac{h_1 Q_1}{p_1}+ c_1 (1-\gamma) \pi_1 \mathcal{S}+2 Q-2  \mathcal{S} p_2-\frac{\lambda_1}{\varepsilon}}{  (1-\gamma) \pi_1\mathcal{S}},\\
w_2(\lambda_1)= \displaystyle \frac{\displaystyle \frac{h_2 Q_2}{p_2}+  c_2 (1-\gamma) \pi_2 \mathcal{S}+2 Q-2\mathcal{S} p_2-\frac{\lambda_2}{\varepsilon }}{(1-\gamma)\pi_2\mathcal{S}}, \\
\lambda_1+\lambda_2=\varepsilon \mathcal{S} (r_0^\flat-r_0).
\end{array}
\right.
\end{equation*}
We have $\lambda_1+\lambda_2>0$ if and only if $r_0<r_0^\flat$. Calculating the respective values of $\varphi_{1}(\lambda_1)$ = $\mathcal{R}_{1}^*(w_1(\lambda_1), w_2(\lambda_1))$ and $\varphi_{2}(\lambda_1)$ = $\mathcal{R}_{2}^*(w_1(\lambda_1), w_2(\lambda_1))$ easily shows that $\varphi_{1}$ and $\varphi_{2}$ are linear functions over interval $[0, \varepsilon\mathcal{S}(r_{0}^\flat - r_{0})]$ with respective maximum at $\lambda_1$ = 0 and $\lambda_1$ =  $\varepsilon\mathcal{S}(r_{0}^\flat - r_{0})$ (see Fig. \ref{fig:phi}). Consequently, there cannot be a joint solution that maximizes both $\mathcal{R}_1^*(w_1, w_2)$ and $\mathcal{R}_{2}^*(w_1, w_2)$ unless $r_0 = r_0^\flat$, in which case the optimal point coincides with $\mathbf{w}^C$.
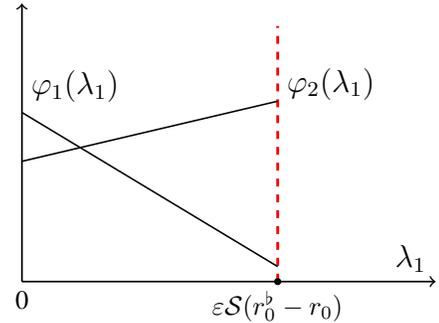
\begin{figure}[!ht]
\centering
\begin{tikzpicture} 
\draw[dashed,line width=1.05 pt,draw=red] (3.4,0) -- (3.4,3.4);
\foreach \coordinate/\label/\pos in { (0,2.6)/{{\large{$\varphi_1 (\lambda_1)$}}}/right, (3.4,2.3)/{\large{$\varphi_2 (\lambda_1)$}}/above right, (3.4,0)/{{$\varepsilon \mathcal{S} (r_0^\flat-r_0)$}}/below,(0,0)/{{0}}/below}
\node[\pos] at \coordinate {\label};
\draw [fill] (3.4,0) circle [radius=0.04]; 
\draw[line width=0.6pt](0,2.25) -- (3.4,0.2) ;
\draw[line width=0.6pt](0,1.6) -- (3.4,2.4) ;
\draw[->, line width=0.64pt] (0.0,0) -- (5.5,0) node[above left] {{\large{$\lambda_1$}}} coordinate(x axis); 
\draw[->, line width=0.64pt] (0,0.0) -- (0,3.7) node[left] {{\large{$$}} }coordinate(y axis); 
\end{tikzpicture} 
\caption{Variations of linear functions $\varphi_1$ and $\varphi_2$ on interval $[0,\varepsilon\mathcal{S}(r_{0}^\flat - r_{0})]$.} 
\label{fig:phi}
\end{figure}

(III) Finally assume $\lambda_1=0$ and $\lambda_2>0$. Solving then each KKT system (\ref{kktPSi}) for 
$i = 1$ and $i = 2$ yields the critical pair ${\bf{w}}^A$ and the value $\lambda_2=\varepsilon \mathcal{S} (r_0^\flat-r_0)$; in particular, we have $\lambda_2>0$ if and only if $r_0<r_0^\flat$. Therefore, ${\bf{w}}^A$ is the pair of optimal wholesale prices when $r_0<r_0^\flat$. Symmetrically, the case $\lambda_1 >$ 0 and $\lambda_2 = $ 0 gives the optimal pair $\textbf{w}^B$ if $r_0<r_0^\flat$. In view of the above properties of functions $\varphi_1$ and $\varphi_2$, however, we easily show that neither $\textbf{w}^A$ nor $\textbf{w}^B$ can be a joint solution that simultaneously maximizes $\mathcal{R}_1^*(w_1, w_2)$ and $\mathcal{R}_2^*(w_1, w_2)$, as required in \eqref{max_11} and \eqref{max_22}. This joint problem has consequently no solution when $r_0<r_0^\flat$.

\end{document}